\documentclass[journal]{IEEEtran}

\usepackage{booktabs} 
\usepackage{color}
\usepackage{amsthm}
\usepackage{ulem}
\usepackage{amsmath,amssymb,amsfonts}
\usepackage{subfigure}
\usepackage{graphicx}
\usepackage{textcomp}
\usepackage{setspace}
\usepackage{mathtools}
\usepackage{cases}
\usepackage{subeqnarray}
\usepackage{makecell}
\usepackage{subfigure}
\usepackage{multirow}

\usepackage{booktabs} 
\usepackage[ruled,linesnumbered,noend]{algorithm2e} 

\newtheorem{thm}{Theorem}

\def \ie{i.e., }
\def \eg{e.g., }
\def \fig{Fig. }
\def \eq{Eq. }

\makeatletter
\def\bstctlcite{\@ifnextchar[{\@bstctlcite}{\@bstctlcite[@auxout]}}
\def\@bstctlcite[#1]#2{\@bsphack
  \@for\@citeb:=#2\do{%
    \edef\@citeb{\expandafter\@firstofone\@citeb}%
    \if@filesw\immediate\write\csname #1\endcsname{\string\citation{\@citeb}}\fi}%
  \@esphack}
\makeatother

\makeatletter
\long\def\@makecaption#1#2{\ifx\@captype\@IEEEtablestring%
\footnotesize\begin{center}{\normalfont\footnotesize #1}\\
{\normalfont\footnotesize\scshape #2}\end{center}%
\@IEEEtablecaptionsepspace
\else
\@IEEEfigurecaptionsepspace
\setbox\@tempboxa\hbox{\normalfont\footnotesize {#1.}~~ #2}%
\ifdim \wd\@tempboxa >\hsize%
\setbox\@tempboxa\hbox{\normalfont\footnotesize {#1.}~~ }%
\parbox[t]{\hsize}{\normalfont\footnotesize \noindent\unhbox\@tempboxa#2}%
\else
\hbox to\hsize{\normalfont\footnotesize\hfil\box\@tempboxa\hfil}\fi\fi}
\makeatother

\begin{document}
\pagenumbering{gobble}
This work has been submitted to the IEEE for possible publication. Copyright may be transferred without notice, after which this version may no longer be accessible.
\newpage
\pagestyle{empty}

\bstctlcite{BSTcontrol}

\pagenumbering{arabic}
\title{ArSMART: An Improved SMART NoC Design Supporting Arbitrary-Turn Transmission}

\author{
        Hui~Chen,
        Peng~Chen,
        Jun~Zhou,
        Duong~H.~K.~Luan,
        and~Weichen~Liu

\thanks{W. Liu, H. Chen, J. Zhou and D. Luan are with the School of Computer Science and Engineering, Nanyang Technological University, Singapore. E-mail: (\{hui.chen, liu\}@ntu.edu.sg).}
\thanks{P. Chen is with the School of Computer Science and Engineering, Nanyang Technological University, Singapore,
and also with the College of Computer Science, Chongqing University, Chongqing, China.}}

\maketitle

\begin{abstract}
SMART NoC, which transmits unconflicted flits to distant processing elements (PEs) in one cycle through the express bypass, is a high-performance NoC design proposed recently. 
However, if contention occurs, flits with low priority would not only be buffered but also could not fully utilize bypass. 
Although there exist several routing algorithms that decrease contentions by rounding busy routers and links, they cannot be directly applicable to SMART since it lacks the support for arbitrary-turn (\ie the number and direction of turns are free of constraints) routing.
Thus, in this article, to minimize contentions and further utilize bypass, we propose an improved SMART NoC, called ArSMART, in which arbitrary-turn transmission is enabled.
Specifically, ArSMART divides the whole NoC into multiple clusters where the route computation is conducted by the cluster controller and the data forwarding is performed by the bufferless reconfigurable router. 
Since the long-range transmission in SMART NoC needs to bypass the intermediate arbitration, to enable this feature, we directly configure the input and output ports connection rather than apply hop-by-hop table-based arbitration.  
To further explore the higher communication capabilities, effective adaptive routing algorithms that are compatible with ArSMART are proposed. 
The route computation overhead, one of the main concerns for adaptive routing algorithms, is hidden by our carefully designed control mechanism.
Compared with the state-of-the-art SMART NoC, the experimental results demonstrate an average reduction of 40.7\% in application schedule length and 29.7\% in energy consumption.

\end{abstract}

\begin{IEEEkeywords}
SMART NoC, arbitrary-turn transmission, contention-minimized routing, bypassing, end-to-end latency.
\end{IEEEkeywords}

\section{Introduction}
\label{section:1-introduction}
\IEEEPARstart{W}{ith} the increasing number of processing elements (PEs) integrated into one chip, the communication between PEs becomes the bottleneck for performance improvement. 
Based on the modified Amdahl's law \cite{yavits2014effect} which considers the effect of communication and synchronization in multi-core systems, the communication bottleneck damps the speedup gained by parallelism and computation acceleration. 
To support high-speed communication among PEs, network-on-chip (NoC), as a widespread communication infrastructure for large-scale many-core systems, has been refined and evolved in recent works. 
SMART NoC \cite{krishna2013breaking}, which transmits unconflicted flits to distant PEs within one cycle through express long-distance bypass paths, is one of the most successful NoC designs. 
Experiments \cite{krishna2013breaking} show that if every flit is magically sent from the source to its destination by using the single-cycle long-distance path, up to 85\% application schedule length reduction can be achieved compared with state-of-the-art traditional NoCs. 
This is the ``ideal" performance that SMART provides, with the optimistic assumption of single-cycle source-destination paths for all flits.

However, in practice, the actual SMART NoC performance is far away from the ideal case since the single-cycle long-distance path can hardly be built for all flits since only the winner of the arbitration among multiple long path setup requests can set up long-range links. 
If one packet is blocked by other packets, its bypass is broken which degrades the benefits gained by SMART NoC. 
Besides, the long-range path establishment is costly due to additional pipeline stages and broadcast links. 
To reduce wire and energy overhead of original SMART NoC, novel designs \cite{chen2016reducing,perez2019smart++,SHARPNoC} are proposed.
Also, researchers try to reduce contentions from the task mapping \cite{yang2017task} and routing \cite{9045103} perspectives. 
The first work turns to task mapping which is limited by the availability of PEs and only performs well in homogeneous systems.
Peng et al. \cite{9045103} try to avoid contention through XY-YX routing with intermediate nodes. 
However, in such design, routes for messages are not fully flexible and constrained by the number of turns.
Thus, the contention issue is not fully addressed in aforementioned works.
A straightforward way to significantly reduce the contentions is to relax these routing constraints and enable the data transmission of arbitrary-turn paths.

The challenge for SMART NoC to support arbitrary-turn transmission is placed by its distributed decision-making mechanism. 
In the start router, the route for a packet is locally computed and then a SMART-hop setup request (SSR), which carries the route information, is broadcast to the downstream routers via dedicated repeated wires to establish bypass. 
This local decision-making mechanism limits the routing algorithm used in SMART NoC in two aspects. 
(\romannumeral1) With the limited area constraint and deadlock requirement, the current route computation module within the SMART NoC router is rather functionally limited, resulting in that only rule-based routing strategy (\eg XY), which is deterministic and only allows specific turns, is applied. 
(\romannumeral2) The SSR delivery is constrained by the dedicated wires or using specific SSR network \cite{chen2016reducing}, which does not support SSR transmission with arbitrary-turn.
Thus, even if we revise the original route compute unit and let it support arbitrary-turn transmission, e.g., using the table-based method, the constrained SSR delivery is not compatible with the arbitrary-turn transmission. 
To support the single-cycle long-distance transmission with arbitrary-turn, centralized or cluster-based design is needed. 
Also, inspired by that the optimal solution is easier to be derived based on global information instead of local information, the centralized or cluster-based method could manage NoC resource (i.e., routers and links) better. 

In this article, we propose a novel NoC design based on SMART NoC \cite{krishna2013breaking}, called ArSMART, which significantly decreases resource contentions and further fully utilizes bypass via our proposed mechanism of establishing arbitrary-turn paths.
The main contributions of our article are as follows: 
\begin{itemize}
  \item [1)]
  We develop an NoC design, ArSMART NoC, to set up single-cycle long-distance paths and support arbitrary-turn data transmission, which significantly reduces resource contentions.
  Specifically, ArSMART divides the whole NoC into multiple clusters where the route computation is conducted by the cluster controller and the data forwarding is performed by the bufferless reconfigurable router. 
  \item [2)]
  We present corresponding routing algorithms that enable ArSMART to manage NoC resources efficiently. 
 Specifically, we conduct the route computation to generate a route before they demand at runtime, considering the real-time network state.
 The challenge to design routing algorithms for ArSMART is the difference of network states used in route computation and actual transmission. 
 Our algorithms manage to minimize such impact and lessen contentions to improve NoC performance.
  \item [3)]
  We implement the ArSMART design and matched routing algorithms in Gem5 \cite{binkert2011gem5}, and conduct a full system simulation to show their effectiveness. 
  Compared with the state-of-the-art SMART NoC, the experimental results demonstrate an average reduction of 40.7\% in application schedule length and 29.7\% in energy consumption.

\end{itemize}

The rest of this article is organized as follows: Section \ref{section:motivation} provides examples to illustrate our motivations. Section \ref{section:Definition} summarizes the notations we used in this article and presents the problem definition. The details of our design together with the proposed routing algorithms for different cases are shown in Section \ref{section:3-architecture}. To prove the efficiency of our proposed design, evaluations on performance, area and power are presented in Section \ref{section:5.2-experiment}. Finally, Section \ref{section:2-relatedwork} discusses related works and Section \ref{section:6-conclusion} concludes the article.
\begin{figure*}[t!]
\centering
\includegraphics[width=7.2in]{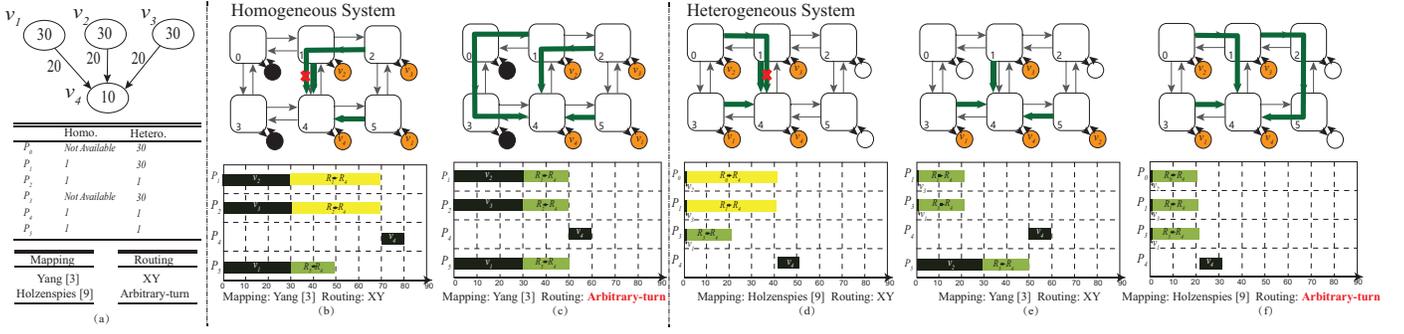}
\vspace{-25pt}
\caption{Motivation examples. (a). DAG modeled application and processing rate of different PEs; (b). Communication-aware mapping and XY routing in homogeneous system; (c). Communication-aware mapping and arbitrary-turn routing in homogeneous system; (d). Computation-aware mapping and XY routing in heterogeneous system; (e). Communication-aware mapping and XY routing in heterogeneous system; (f). Computation-aware mapping and arbitrary-turn routing in heterogeneous system.}
\label{figure:example}
\end{figure*}
\begin{center}
\begin{figure}
\includegraphics[width=3.5in]{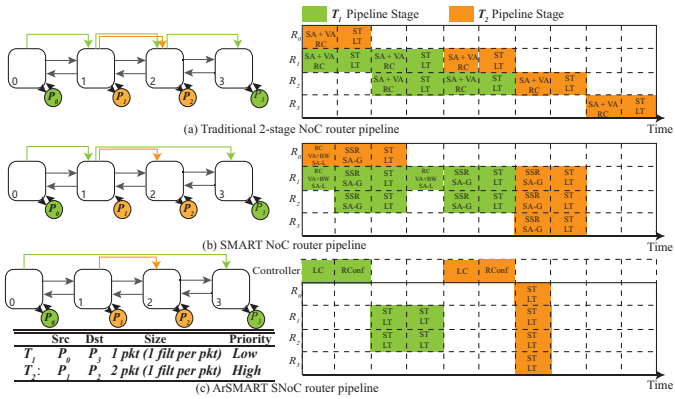}
\caption{Illustration of traditional, SMART and ArSMART NoC timeline.}
\label{figure:timeline}
\end{figure}
\end{center}
\section{Motivation}
\label{section:motivation}
In this section, we motivate the benefits of supporting arbitrary-turn transmission and cluster-based resource management through the following examples. 

Compared with XY routing applied in SMART NoC, arbitrary-turn transmission can fully utilize NoC resources under the same mapping strategy.
Given the task graph and its mapping in \fig \ref{figure:example}(a), the given application is represented as a directed acyclic graph (DAG). 
For each node $v \in \mathcal{V}$, its task workload is indicated using the number inside the node, and for each edge $e_{u, v} \in \mathcal{E}$ from task $u$ to task $v$, its message size is represented by the number beside the edge. 
The processing rate of different PEs is listed in the table of \fig \ref{figure:example}(a). 
Researchers proposed communication-aware task mapping algorithms \cite{yang2017task} to minimize contentions, in which up to 44.1\% improvement in communication efficiency can be achieved by minimizing contention for SMART NoC. 
However, even cooperated with this task mapping algorithm, the XY routing would encounter contentions in \fig \ref{figure:architecture}(a) due to the limitation of PEs' availability. Totally, 80 time units are consumed as shown in \fig \ref{figure:example}(b). 
If arbitrary-turn routing is applied, only 60 time units are needed in \fig \ref{figure:example}(c). 

For the heterogeneous system, arbitrary-turn routing algorithms can cooperate with computation-aware mapping to get the optimal performance for both computation and communication. 
With the same task graph in the previous example, we apply two mapping algorithms, the communication-aware mapping algorithm proposed in \cite{yang2017task} and the computation-aware mapping presented in \cite{holzenspies2010run}. 
Under the XY routing, when applying computation-aware mapping, the timeline is shown in \fig \ref{figure:example}(d). 
Due to the contention, the total schedule length is 51 time units. 
As shown in \fig \ref{figure:example}(e), if the communication-aware mapping algorithm and XY routing are applied, even if no contention occurs, the total schedule length is 60 time units due to its prolonged task execution time. 
However, if the proposed arbitrary-turn routing and computation-aware mapping are applied, the schedule length is reduced to 31 time units, as indicated in \fig \ref{figure:example}(f). 

Our arbitrary-turn routing design applies cluster-based resource management and removes per-router arbitration. 
To guarantee there is no contention during the transmission, ArSMART NoC blocks low-priority messages at the source. 
The benefit we can gain from such design is shown in \fig \ref{figure:timeline}. 
Generally, the traditional NoC router processes each flit through 5 stages \cite{dally2004principles}: route computation (RC), virtual channel allocation (VA), switch allocation (SA), switch transmission (ST) and link transmission(LT). The state-of-the-art research \cite{park2012approaching} shows that these 5 stages can be pipelined as shown in Fig. \ref{figure:timeline}(a).
As illustrated in Fig. \ref{figure:timeline}(b), the SA stage in the SMART router contains two steps: switch allocation local (SA-L) and switch allocation global (SA-G). 
Given two messages $T_1$ and $T_2$ and their information as listed in Fig. \ref{figure:timeline}. Since $T_1$ can bypass Router 2 in SMART NoC, the total transmission time is shortened compared with traditional NoC. 
However, since the bypass of $T_1$ is interrupted by $T_2$, the bypass is broken and needs to be set up again. 
If we force $T_1$ to wait at the source, both messages can benefit from the long-range path, and the total time is decreased to 7 cycles. 
Moreover, we note that this is an extreme example, \eg the message size is no more than 2 packets. 
If the path can be used for the messages consisting of more packets, the path configuration overhead is shared further. 
\vspace{-10pt}
\section{Problem Definition}
\label{section:Definition}
In this section, we will define the problem and the objective this article targets. 

\subsubsection{Application}
An application is represented by a directed acyclic graph (DAG) $G=(\mathcal{V}, \mathcal{E})$, where $\mathcal{V}$ is the set of the computation tasks and $\mathcal{E} \subseteq \mathcal{V} \times \mathcal{V}$ is the set of data transmission between tasks. 
For each node $v \in \mathcal{V}$, the task workload is represented by $w_{v}$, and for each edge $e_{u, v} \in \mathcal{E}$ from $u$ to $v$, its message is represented by $m_{e}$. 
The set of all messages is notated using $\mathcal M$. Each message $m_e = \left\{p_1,p_2,...,p_i,...,p_j\right\}$ consists of $j$ packets, and each packet $p_i =\left\{f_1,f_2,...,f_i,...,f_k\right\}$ consists of $k$ flits. 

\subsubsection{Architecture}
Formally, the 2D mesh NoC-based SoC is formed of $N \times N$ PEs and routers. 
The PE in $i^{th}$ row and $j^{th}$ column is denoted by $c_{ij}$ ($(c_{ij} \in \mathcal{C}$ and $\mathcal{C}=\left\{c_{1,1}, c_{1,2}, \ldots, c_{N,N}\right\})$). 
The processing rate of $c_{ij}$ is denoted by ${s_{c_{ij}}}$.   
The router in $i^{th}$ row and $j^{th}$ column is denoted by $r_{ij}\left(r_{ij} \in \mathcal R\right.$ and $\left.\mathcal R=\left\{r_{1,1}, r_{1,2}, \ldots, r_{N,N}\right\}\right)$.

\subsubsection{Mapping Algorithm}
Application mapping $\mathcal{F}$ is a function from tasks $\mathcal{V}$ to processors $\mathcal {C}$. $\mathcal{F}(v)=c$ represents the mapping of task $v$ onto processor $c$. 
Based on $w_v$ and $s_c$, the execution time $q_{v,c}$ for task $v$ on processor $c$ is estimated by $q_{v,c}=w_v/s_c$. 
We note that, due to the existence of branch operations, the $q_{v,c}$ value, estimated in the design time, may not be the same as the actual execution time in the run time. 

\subsubsection{Routing Algorithm}
The routing algorithm is a function from messages $\mathcal {M}$ to routers $\mathcal {R}$. 
$\mathcal{G}(m)=\gamma$ represents message $m$ transmitting over the route $\gamma$. 
The route $\gamma = \{r_{src}, \ldots, r_{dst}\}$ is a set of routers that forward this message from the source to the destination. 
In the distributed NoC system, to reduce the route computation overhead and avoid deadlock, constraints are added on the route, \eg flits need to traverse in X direction at first and then Y direction for XY routing. 
We use notations $\mathcal G_w$ and $\mathcal G_{w/o}$ to represent the routing algorithms with or without constraints. 
 Also, we use the item ``arbitrary-turn route'' to notate the routing algorithm without any constraints.

\subsubsection{Packet End-to-End Latency}
For the distributed NoC, the end-to-end latency ${L}_{e2e}$ of a packet consists of head flit transmission latency ${L}_\textit{head}$, serialization latency ${L}_\textit{seri}$ and contention latency ${L}_\textit{ct}$, as shown in \eq \ref{E.q.:end2end}.
\begin{equation}
  \begin{split}
{L}_{e2e} = {L}_\textit{head} + {L}_\textit{seri} + {L}_{ct}
\label{E.q.:end2end}
  \end{split}
\end{equation}
Generally, in traditional hop-by-hop traversal NoCs, flits are forwarded hop by hop. The end-to-end latency of a packet for traditional NoCs ${L}_{e2e}^{C}$ can be formulated as follows. 
\begin{equation}
  \begin{split}
{L}_{e2e}^{C} = L_r \cdot (|\gamma|-1) + L_w \cdot |\gamma| + L_w \cdot (|f|-1) + {L}_{ct}
\label{E.q.:traditionalLatency}
  \end{split}
\end{equation}
As shown in \eq \ref{E.q.:traditionalLatency}, $L_r$ and $L_w$ are the router-stage delay and propagation delay between two adjacent routers, respectively. 
$|\gamma|$ represents the number of routers of the route from the source to the destination; $|f|$ refers to the number of flits of a packet. 
For SMART NoC, due to the bypassing of intermediate routers, the end-to-end latency ${L}_{e2e}^{S}$ is represented by \eq \ref{E.q.:SMARTLatency}. 
Where $|ct|$ and $|limit|$ are the bypass broken overhead suffered from contention and limitation of $\textit{HPC}_\textit{max}$.  
From the formula, we find that the SMART latency is affected by contention in two aspects: the extra blocking latency and the prolonged head flit transmission latency caused by bypass break. 
\begin{equation}
\label{E.q.:SMARTLatency}
\begin{aligned}
{L}_{e2e}^{S}
&= 2 \times\left(L_{r}+L_{w}\right)+(|ct|+|limit|) \times\left(L_{r}+L_{w}\right) \\
&+\left(|f|-1\right) \times L_{w}+ {L}_{ct}
\end{aligned}
\end{equation}
For our proposed ArSMART NoC, since the route is configured by the controller directly, head flit carrying the route information is unnecessary, which means ${L}_{head}$ used to set up route is replaced with route configuration time, ${L}_{conf}$. 
Then, the data can be transmitted from the source to the destination costing ${L}_\textit{tr}$. 
For the contention delay, our method eliminates the contention at intermediate routers. 
Instead, an additional delay at the source ${L}_{cs}$ is added. The latency to transmit a message using our design ${L}_{m}^{ar}$ is represented by \eq \ref{E.q.:messageLatency}. We denote the latency without contention using ${L}_{w/oc}$.
\begin{equation}
\label{E.q.:messageLatency}
\begin{aligned}
{L}_{m}^{ar} = {L}_\textit{conf} + {L}_\textit{tr} + {L}_\textit{cs} = {L}_\textit{w/oc}+{L}_\textit{cs}
\end{aligned}
\end{equation}
Note that in our design, the path is built at message level rather than packet level. 
Thus, the configuration time is shared by multiple packets as ${L}_{conf} / |p|$. 
Finally, the end-to-end latency of a packet for our design ${L}_{e2e}^{ar}$ is represented by \eq (\ref{E.q.:ourLatency}).
\begin{equation}
\label{E.q.:ourLatency}
\begin{aligned}
{L}_{e2e}^{ar} &= {L}_\textit{conf}/|p| + |limit| \times\left(L_{r}+L_{w}\right)\\
&+\left(|f|-1\right) \times L_{w} + {L}_\textit{cs}/|p|
\end{aligned}
\end{equation}
\begin{table}
    \caption{Notations Used in This article}
    \renewcommand{\arraystretch}{1.0}
    \arrayrulewidth=0.85pt
    \tabcolsep 3pt
    \centering
    \footnotesize
    \begin{tabular}[c]{c  c}
        \hline
        \hline
        {\bf Notation }&\bf{Description}\\
        \hline
        $\mathcal{V}$, $\mathcal{E}$  $\mathcal{M}$ &  The set of tasks nodes, edges and messages. \\
        $m,p,l$ & The notation of message, packet and flit. \\
        $|.|$ & The number of elements.\\
        $\mathcal{C}$, $\mathcal {R}$& The set of processing elements and routers. \\
        $r_{xy}$& The router in $x^{th}$ row and $y^{th}$ column. \\
        $q_{v,c}$& The execution time of $v$ on $c$. \\
        $\mathcal{F},\mathcal{G}$& The mapping and routing algorithm. \\
        $\gamma_m$& The route to transmit data of message $m$. \\
        ${L}$ & The latency related to transmission. \\
        $\tau_{m}$ & The priority of message $m$.\\
        \hline
        
    \hline
    \end{tabular}
    \label{notations}

\end{table}
In Table \ref{notations}, we summarize notations we used throughout this article. We have presented message latency of our design in \ref{E.q.:messageLatency}. 

The objective of our designs is to minimize the latency for each message, \ie $Min({L}_{m}^{ar})$. Specifically, using efficient control mechanism, we tried to minimize ${L}_\textit{conf}$. The objective of our routing algorithm is to find the route which minimize the ${L}_\textit{cs}$, \ie $\mathop{\arg\min}\limits_{\gamma}{L}_{cs}$

\begin{center}
\begin{figure}
\includegraphics[width=3.5in]{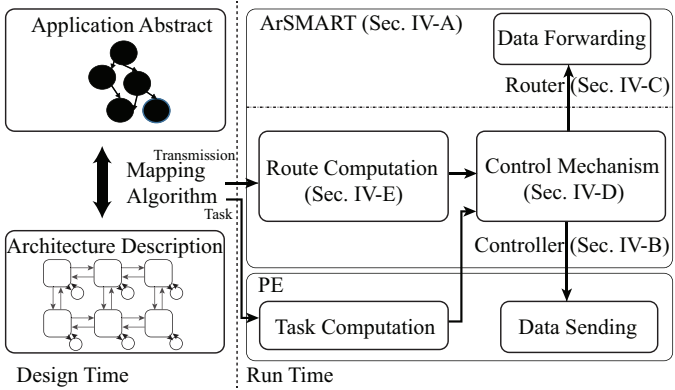}
\caption{Illustration of ArSMART NoC context.}
\label{figure:context}
\vspace{-15pt}
\end{figure}
\end{center}

\begin{figure*}[t!]
\centering
\includegraphics[width=7.4in]{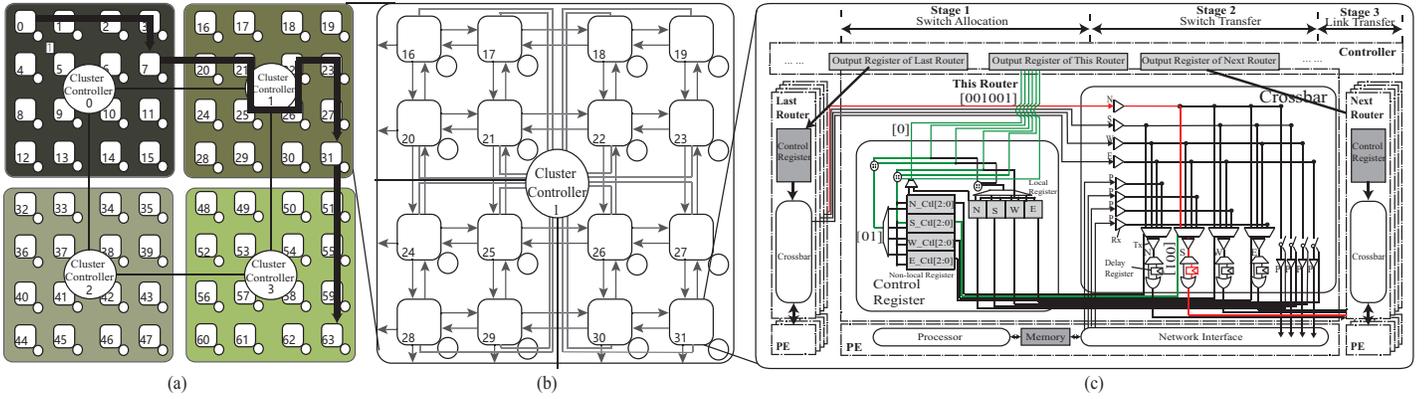}
\vspace{-20pt}
\caption{ArSMART NoC Design (a). Overview of ArSMART; (b). Cluster structure; (c). Router design.}
\label{figure:architecture}
\end{figure*}
\section{Proposed ArSMART NoC}
\label{section:3-architecture}
The hardware-software co-design flow of ArSMART NoC is presented in Fig. \ref{figure:context}. 
For a given application abstract and specific architecture description, the mapping algorithm decides the piece of code every PE should execute. 
After task mapping, the task graph integrated with mapping information is generated. 
Such task graph briefly describes task information which includes the codes' partition and location as well as transmission information which describes the source and destination of messages. 
With task and transmission information, the routing computation and task computation can be conducted concurrently to cover the route computation overhead of adaptive routing algorithms. 

The ArSMART NoC mainly consists of two components, router and controller. 
Assisted by the control mechanism, the configuration of each router generated in the controller can be accurately executed in routers, then the data from the source can be forwarded to the destination precisely.
In following subsections, we will detail our proposed design.
\vspace{-10pt}
\subsection{Design Overview}
Fig. \ref{figure:architecture} demonstrates our proposed NoC.
The whole NoC is separated into multiple clusters which consist of several routers and one cluster controller.
An illustrative example of 16 routers in one cluster is given in Fig. \ref{figure:architecture}(a).
The cluster controller connects to every router within the cluster using a point-to-point control link as shown in Fig. \ref{figure:architecture}(b). 
To configure inter-cluster transmission paths, each cluster controller is connected to its adjacent controllers by 16-bit wires.
The link state within one cluster is collected by the cluster controller.

We note that ArSMART can be scalable to any mesh-size NoCs by applying proper cluster size and the number of clusters. 
However, the cluster size is limited due to the control signal distribution and transmissions it can process. 
The maximum distance can be traversed within one cycle is limited, i.e., 8 mm at 1 GHz \cite{krishna2013breaking}. 
If the cluster controller is placed at the center, the maximum cluster size is 8$\times$8. 
If the memory size for the controller is 10 MB and each thread consumes 10 KB, the total number of messages the controller can process is 1024, which is enough for 64 PEs. 

The main process to transmit a message in ArSMART is summarized as follows. After tasks are mapped to processors, the route for one message is computed.
When the route for this communication request is demanded and allowed by all required controllers, the cluster controller configures the corresponding routers directly. Then the single-cycle multi-hop bypass path is established successfully. After the transmission, the path is released and the corresponding link state is updated as free. 
No local arbitration is needed during this process.
\vspace{-10pt}
\subsection{Controller Design}
\vspace{-15pt}
\begin{center}
\begin{figure}
\includegraphics[width=3.5in]{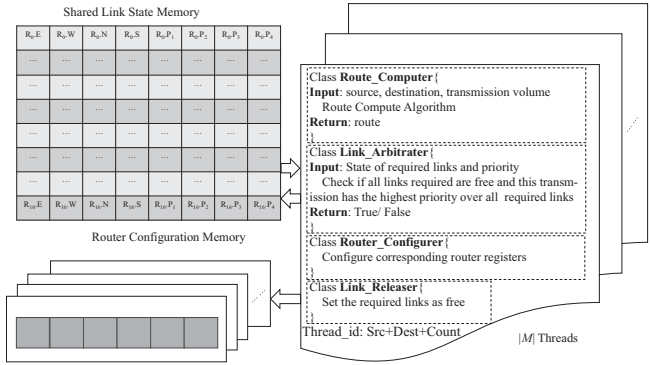}
\vspace{-20pt}
\caption{Illustration of controller.}
\label{figure:controller}
\vspace{-10pt}
\end{figure}
\end{center}
The controller is responsible for route computing, link arbitration and link state updating.
In this article, we do not limit the specific implementation of the controller. 
One possible solution is that the controller is one of the PEs. 
Fig. \ref{figure:controller} demonstrates our software design for the controller. 

For each message, there is one thread responsible for it. 
The thread id is the combination of the source, destination and the number of messages sending from the same router. 
Four main functions of a thread are: \emph{route computation}, \emph{link arbitration}, \emph{router configuration} and \emph{link release}. 
(\romannumeral1) Based on the source and destination information, the thread computes the route for this message. 
(\romannumeral2) After task execution finishes, the thread checks whether this message has the highest priority among all requested links. 
(\romannumeral3) If the checking result is true, this thread configures the routers based on the computed route. 
(\romannumeral4) When communication finishes, the thread updates the corresponding link state. Details of these functions we will illustrate in the following sections.

The controller has a shared memory that stores the link state (i.e., busy/free). The size of shared memory equals to $|R_n| \times |ports|$, where $|R_n|$ is the number of routers in this cluster and $|ports|$ is the number of ports for each router. 
To ensure that every entry of shared memory in the controller is accessed by at most one thread simultaneously, we present our synchronization mechanisms (i.e., mutual exclusion). 
If one thread wants to occupy a link and change the link state from free to busy, it should win the priority arbitration. To perform the priority arbitration, we create a priority queue for each link. 
Every message requests for one link is inserted into the queue base on its priority. 
In our design, we apply the first-come-first-serve policy, which means the communication firstly requesting the link has the highest priority. 
When priority arbitration is performed, the corresponding priority queues return the first item as the result.
Since only one message would win the priority arbitration for one link, the link state is updated by one thread. 
For low overhead, the resource arbitration is non-preemptive so that link arbitration is not conducted for one link if the required link is taken by another message already.
After the message finishes, the thread would update the link state from busy to free. 
Since the link is occupied by one message, only one thread would change the value of the link state. 
We note that the cluster controller connects to every router in its cluster. 
Thus, configuration information can be sent to corresponding routers simultaneously.
\subsection{Configurable Router Design}
\begin{figure}
\centering
\includegraphics[width=3.1in]{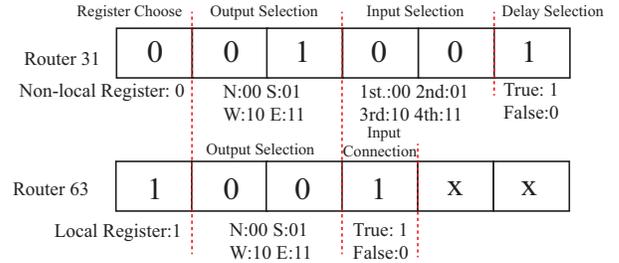}
\caption{Configuration Decoding.}

\label{figure:signalForm}
\vspace{-15pt}
\end{figure}
The router design is detailed in Fig. \ref{figure:architecture}(c). The 6 bits configuration signal, \textit{router-configure}, sent from controller are decoded and stored in corresponding registers as illustrated using the green lines in Fig. \ref{figure:architecture}(c).
The configuration register consists of two separate registers. 
One register with 4 entries (3 bits for each entry), named non-local register, links the input port(s) to non-local ports (\ie north, south, east, west ports). 
The other register with 4 entries (1 bit for each entry). named local register, configures the input port(s) to the local processor.
ArSMART uses a 6-bit control signal to change the content of these two registers. 
One of the two registers is selected by the first bit.
If the non-local register is chosen, the following 2 bits are used to select 1 entry and the last 3 bits are stored in that entry. 
The last 1 bit is used to control the delay register. 
Otherwise, the following 2 bits are used to choose 1 entry and only the next 1 bit is stored in that entry. 
The processors can send and receive data from all directions simultaneously since it connects to all non-local ports. 
To decrease wire delay \cite{krishna2013breaking}, the Rx and Tx asynchronous repeaters are used. Considering the limitation of $\textit{HPC}_\textit{max}$, ArSMART deploys delay registers with the size of 1 flit in each input port to temporarily hold data.

In the example of transmitting data from 0 to 63, we configure router 31 and let it temporarily hold the data in delay register and forward data from N port to S port in the next cycle, as shown using the red line in Fig. \ref{figure:architecture}(c). The entry 01 of the non-local register should be configured as 001, as indicated in Fig. \ref{figure:signalForm}.
Another example is the configuration of router 63. 
By setting the first bit to 1, we select the local register. 
Since the connection between N port and local port should be built, the entry for N port is chosen and set as 1, as shown in Fig. \ref{figure:signalForm}.
With this configuration coding, any connection between input ports and output ports can be established. 

\begin{figure}[h]
\centering
\includegraphics[width=3.1in]{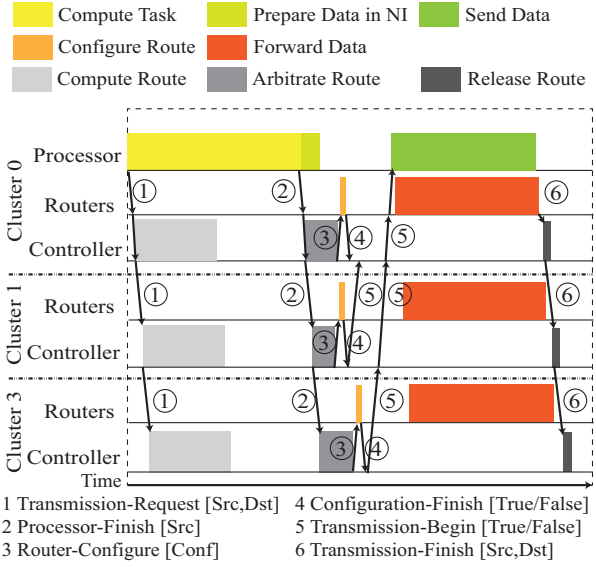}

\caption{Control mechanism for path establishment and release.}
\vspace{-10pt}
\label{figure:communication}
\end{figure}

\subsection{Control Mechanism}
We propose a C4R (\ie compute, check, configure, communicate and release operations) mechanism to support arbitrary-turn transmission without any local decision making. 
In the following, we will introduce such mechanism. 
Without loss of generality, we use Fig. \ref{figure:communication} to detail the control process to transmit the message from Router 0 to 63 as illustrated in Fig. \ref{figure:architecture}(a). 

\noindent  $\bullet$
\textbf{Compute.} 
This process computes the route for a message. 
After task mapping, the source and destination of messages are known.
At the beginning of the task execution, for one message generated from this task, the \textit{transmission-request} signal (12 bits) with the source and destination information is sent to the cluster controller. 
If more than one transmission request is submitted, the controller will process them using multiple threads. 
Task computation and route computation are conducted at the same time, hiding the route computation overhead. 
Time for route computation may be longer than task execution time, which will be discussed in the next section.
In the example in Fig. \ref{figure:architecture}(a), the destination is beyond the control of its cluster controller. 
The first controller sets a ``temporary destination'' at the boundary of the cluster and sends the communication request whose ``temporary source'' is this ``temporary destination'' to another cluster controller as illustrated in Fig. \ref{figure:communication}. 
This step would be conducted until the ``temporary source'' and destination are in the same cluster.
The ``temporary destination'' is chosen randomly among the available boundary routers $r_b$ with {$r_{{src}_{x}} \leq r_{b_{x}} \leq r_{{dst}_{x}}$ and $r_{{src}_{y}} \leq r_{b_{y}} \leq r_{{dst}_{y}}$}. In the example, router 7 and 31 are boundary routers.
Considering the limitation of $\textit{HPC}_\textit{max}$, flits would latch router(s) whose hop count from the source or ``temporary source'' is multiple times of $\textit{HPC}_\textit{max}$ and continue to transmit in the next cycle.
Due to the concern about $\textit{HPC}_\textit{max}$, filts should latch boundary routers to let the count of hops for other clusters start from 0.

\noindent  $\bullet$ \textbf{Check.}
A message is eligible to transmit if and only if it can win all required links. 
The processor sends the \textit{processor-finish} (6 bits) which includes the source information to the controller as long as it finishes task execution. 
The transmission path for one message is exclusive, meaning that this path cannot be used by other messages simultaneously. 
For one message, the controller checks whether this message has the highest priority among all requested links.
If the transmission is beyond the cluster, the cluster needs to forward the \textit{processor-finish} request to other corresponding clusters.
The controller only checks links within its cluster. In \fig \ref{figure:architecture}(a), controller 0 checks links used to transmit data from router 0 to 7 and the link connected to the east port of router 7.
Note that if one message finds one link it needs is taken by another message, it would not request any links along its route until that message ends the transmission and release the path.

\noindent  $\bullet$ \textbf{Configure.}
If all required links are available after the ``check'' process, routers along the assigned path are configured by \textit{router-configure} signal (6 bits) which we have discussed previously. 
For the inter-cluster case, the corresponding cluster controller(s) would configure routers if the checking result is true.
Since the cluster controller connects with all routers, it can send the configuration information to all related routers simultaneously. In \fig \ref{figure:architecture}(a), controller 1 gets the true result at first, and then it configures router 20, 21, 25, 26, 22, 23, 27, 31 within one cycle.
After routers finish the configuration, they send the \textit{configuration-finish} signal (1 bits) to its cluster controller. 
We note that the case that more than one message configures the same router at the same time exists. 
However, such additional delay caused by router configuration is limited since this case rarely happens. Also, since there are 5 input ports and 5 output ports only, at most 5 cycles are needed to configure a router.
Updating the link state promptly is required in our system. 
Before routers are configured by the controller, corresponding link states are updated as busy, so that the other messages cannot transfer data along this path.

\noindent  $\bullet$ \textbf{Communicate.}
After routers within the cluster are configured correctly, the controller sends the \textit{transmission-begin} signal (1 bits). 
In the inter-cluster case, the local cluster would send the \textit{transmission-begin} signal to the source cluster. The source cluster sends the \textit{transmission-begin} signal to the processor if and only if it collects all \textit{transmission-begin} signals from required clusters. In \fig \ref{figure:architecture}(a), the controller 0 sents the \textit{transmission-begin} signal to the source PE after collecting \textit{transmission-begin} signals from controller 1 and 3. 
Then the source PE begins to send data.

\noindent  $\bullet$ \textbf{Release.}
After router 0 transmitting the last flit, the links from router 0 to router 63 should be released. 
In the first cycle, router 0 sends the \textit{transmission-finish} signal (12 bits) with the source and destination of this message to cluster controller. 
After the cluster controller receives this signal, it checks this transmission beyond this cluster so it forwards the \textit{transmission-finish} to other corresponding cluster controllers (cluster 1 and cluster 3).
Finally, all related controllers update their link state correctly.

In our system, since the route is configured by the controller directly, head flit is unnecessary, which means ${L}_{head}$ used to set up the path is replaced by route configuration time, ${L}_{conf}$. 
The route configuration time is $Max(0,2 \times( L_{cn}+L_{rc})+L_{rls}-L_{pre})$, where $L_{cn}$ is delay cycles to coordinate clusters this path involves in; $L_{rc}$ refers to delay cycles caused by router configuration and release. The maximum $L_{rc}$ is 5 as we discussed before; $L_{pre}$ is the data preparation cycles in NI. ${L}_{conf}$ can be overlapped by $L_{pre}$. For a given path, the maximum configuration time can be computed by: $2 \times( |cn|+5)+|cn|$, where $|cn|$ is the number of clusters this path involves in.

\vspace{-10pt}
\subsection{Routing Algorithm}
We have presented our C4R transmission mechanism in the previous section. 
To fully utilize the NoC resources, we propose the corresponding routing algorithms. 
As mentioned before, we advance the route computation to the start of task execution to cover the route computation overhead. 
Note that, if too many messages are added to the task graph, the competition for the resource is inevitable.
Instead of competing at intermediate routers, blocking the low-priority transmission at the source can separate different transmissions without any extra delay. We try to decrease such blocking delay at the source using our routing algorithms.

Generally, due to the existence of branch operations, the task execution time is unknown beforehand, as described in \cite{ravindran2014scheduling}, and only the source and destination of messages are given. 
Suppose we can estimate the message size $|m|$ at the design time, like the case shown in \cite{shen2010dynamic}.
Without knowing the exact transmission start time, we use a greedy strategy to compute the proper route for each message.
The objective of the greedy strategy can be customized to meet different needs.

At the routing algorithm start time point $t$, we use $\mathcal M_t$ to denote the set of messages which have been assigned routes and have not completed their transmission (\ie are transmitting or waiting for transmission). 
In our algorithm, we only consider messages in $M_t$ since the other messages either are uncertain in route so that contention cannot be computed or have finished so that have no influence on the current network state. 

For a message $m_i$, our algorithm try to find a route $\gamma$ which minimizes the estimation upper bound of the blocking latency which suffers from messages in $M_t$. 
In Theorem \ref{them:case1}, we prove that the upper bound of blocking latency at the source $m_i$ suffers from $m_j$ ($j \neq i$, $m_j \in \mathcal M_t$) is $L_{w/oc_{m_j}}$.
\begin{thm}
\label{them:case1}
Given a message $m_i$, the blocking latency at the source that $m_i$ suffers from $m_j$ ($\gamma \cap \gamma_{m_j} \neq \emptyset$, $j \neq i$, $m_j \in \mathcal M_t$) along $\gamma$, $L_{cs_{m_i \leftarrow m_j,\gamma}}$, is upper-bounded by: \\ 
\begin{center}
\vspace{-15pt}
$L_{cs_{m_i \leftarrow m_j,\gamma}} \leq L_{w/oc_{m_j}}$
\end{center}
\end{thm}
\begin{proof}
If $\gamma \cap \gamma_{m_j} = \emptyset$, $m_i$ is not influenced by $m_j$ and $L_{cs_{m_i \leftarrow m_j}} = 0$.\\
(\romannumeral1) $ \tau_{m_j} < \tau_{m_i}$: If $m_j$ starts transmission early than $m_i$, $m_i$ has to wait until $m_j$ ends its transmission and $L_{cs_{m_i \leftarrow m_j,\gamma}} \leq L_{w/oc_{m_{j}}}$. 
If $m_j$ requests links late or at the same time as $m_i$, since $ \tau_{m_j} < \tau_{m_i}$, $m_i$ does not be influenced by $m_j$.\\
(\romannumeral2) $ \tau_{m_j} \geq \tau_{m_i}$: the maximum blocking latency at the source ${m_i}$ suffers from $m_j$ equals $L_{e2e_{m_j}}$, where $L_{e2e_{m_j}} = L_{w/oc_{m_{j}}} + L_{cs_{m_{j}}}$. As illustrated in our control mechanism, if $m_j$ is blocked by another message, it ends the request for all links until that message ends its transmission. 
Thus, such indirect contention does not influence $m_i$ and $L_{cs_{m_{j}}}$ does not lengthen $L_{cs_{m_i \leftarrow m_j,\gamma}}$. Finally, $L_{cs_{m_i \leftarrow m_j,\gamma}} \leq L_{w/oc_{m_{j}}}$.
The theorem is proved.
\end{proof}

\begin{thm}
\label{them:upperbound}
Given a set of messages $\mathcal M_t$, for $m_i$, the estimated delay at the source along the route $\gamma$, $E(L_{cs_{m_i,\gamma}})$  is upper-bounded by: \\ 
\begin{center}
\vspace{-15pt}
$E(L_{cs_{m_i,\gamma}}) \leq \sum_{\forall m_j \in S} L_{w/oc_{m_j}}$, where 
$S=\left\{m_j|\gamma\cap \gamma_{m_j} \neq \emptyset, j \neq i, m_j \in \mathcal M_t\right\}$
\end{center}
\end{thm}
\begin{proof}
$E(L_{cs_{m_i,\gamma}}) = \sum_{\forall m_j \in S}P(m_i \leftarrow m_j,\gamma)L_{cs_{m_i \leftarrow m_j,\gamma}}$, where $P(m_i \leftarrow m_j,\gamma)$ is the probability $m_j$ and $m_i$ compete for the one or more links along $\gamma$ at the same time, $L_{cs_{m_i \leftarrow m_j,\gamma}}$ is the blocking time $m_i$ suffers from $m_j$ along $\gamma$. 
Since $P(m_i \leftarrow m_j,\gamma)\leq1$ and $L_{cs_{m_i \leftarrow m_j,\gamma}} \leq L_{w/oc_{m_j}}$, as illustrated in Theorem \ref{them:case1}, $E(L_{cs_{m_i,\gamma}})$ is upper-bounded by $\sum_{\forall m_j \in S} L_{w/oc_{m_j}}$.
The theorem is proved.
\end{proof}
Based on Theorem \ref{them:upperbound}, the upper bound of delay at source for one message along the route $\gamma$ is the total transmission time without blocking of all messages along this route. 
The transmission time without blocking of message $m$ is $ L_{w/oc} = L_{conf} + {L}_\textit{tr}$, where $ {L}_\textit{tr} \propto |m|$. Since $L_{conf}$ is relatively fixed and small, $ L_{w/oc} \propto |m|$, the size of message $m$. Together with Theorem \ref{them:upperbound}, $E(L_{cs_{m_i,\gamma}}) \propto \sum_{\forall m_j \in S} |m_j|$, where $S=\left\{m_j|\gamma\cap \gamma_{m_j} \neq \emptyset, j \neq i, m_j \in \mathcal M_t\right\}$. As mentioned before, we suppose that the message size $|m|$ is estimated at the design time, like the case shown in \cite{shen2010dynamic}. For a message $m_i$, we compute the $\sum_{\forall m_j \in S }|m_j|$, where $S=\left\{m_j|\gamma\cap \gamma_{m_j} \neq \emptyset, j \neq i, m_j \in \mathcal M_t\right\}$, as the cost of a candidate route $\gamma$.
Among all candidate routes, we choose the route with the minimum cost using Algorithm \ref{alg3}.
Specifically, we initialize the cost matrix in Line 2-4. 
For a vertex in cost matrix, we compute the upper bound of delay estimation for this vertex and its neighbor by accumulating the message size of all messages along this route, as indicated in Line 11. 
Then, the cost matrix is updated by checking whether the total cost can be shortened after adding the delay, shown in Line 12-13. 
The matrix update will end until reaching the destination in Line 8-9. 
Finally, we go through the previous nodes and return the route in Line 14-17.   
\begin{algorithm}
    \caption{Algorithm for General Case}
    \label{alg3}
    \KwIn{An unassigned message $m$\;}
    \KwOut{Route $\gamma$\;}
    create vertex set $Q$\;
    \For{each vertex $v$ in $\mathcal R$}
    {
        $cost_v=inf$;$prev_v=inf$;$Q.enqueue(v)$\;
        
    }
    $cost_s=0$\;
    \While{$Q \neq \phi$}
    {
        $u$ is vertex in $Q$ with the least $cost_u$\;
        $Q.dequeue(u)$\;
        \uIf{$u$ is $destination$}
        {
            break\;
        }
        \For{each neighbor $v$ of $u$}
        {
            $delay = computeCost(v,u)$\; 
            \uIf{$cost_v > cost_u + delay$}
            {
                $cost_v = cost_u + delay; prev_v = u;$
            }
        }
    }
    \uIf{$prev_{dest}$ is defined or $u = source$}
    {
        \While{$u$ is defined}
        {
            $\gamma.push(u);u=prev_u$\;
        }
    }
    Return $\gamma$\;

\end{algorithm}

\noindent  $\bullet$ \textbf{Improved Routing Algorithm for Time-Triggered Case}
\label{timetrigCase}
After task mapping, besides the source, destination and message size, the execution time $q$ can also be known, as described in \cite{silberman1996task}. 
This is common in digital signal processing, 4G and matrix multiplication. 
Also, in the time-triggered real-time system \cite{steiner2010evaluation}, the activities are initiated periodically at predetermined points.
With the execution time $q$ provided, we can compute the route and transmission period for messages accurately by loading the accurate link state of a particular period, as shown in Algorithm \ref{alg1}.
When a task $v$ is mapped onto a core at time point $t$, its messages' initialized start time $t_{start}$ is given in the time-triggered case. If no contention occurs during the transmission, the end time of the message is $t_{start}+L_{w/oc}$. Since $L_{w/oc}= L_{conf} + {L}_\textit{tr}$, where $ L_{conf}$ is upper bounded by $2 \times( |cn|+5)+|cn|$ and ${L}_\textit{tr}$ is computed by $|m|\times|p|(\times|f|-1)\times L_w+|limit|\times(L_r+L_w)$, the upper bound of $L_{w/oc}$, $Max(L_{w/oc})$ can be computed. For one message, we compute the initialized start time and the maximum non-contention end time in Line 1 and Line 2. 
The reservation list is used to save routes for all unfinished messages and their transmission periods.
By checking transmission periods stored in the reservation list, the network state with available links in that period is generated in Line 5.
In this article, we use Dijkstra's algorithm, in Line 6, to find the shortest path within a specific period.
In this period, if data can be transmitted from source to destination in such network graph, route information and this period are inserted into the reservation list in Line 7-9. 
Otherwise, since links are released if and only if the transmission finishes, the algorithm tries the next finishing time point until it finds an available route in Line 11-13.
\begin{algorithm}
\caption{Algorithm for Time-Triggered Case}
\label{alg1}
\KwIn{An unassigned message $m$, route reservation list $RL$\;}
\KwOut{Route $\gamma$ and transmission period [$t_1$,$t_2$]\;}
$t_1 = initial Start Time$\;
$t_2 = t_1 + max(L_{w/oc}) $\;
$\gamma = \phi$\;
\While{(1)}
{
    $\mathcal{G}=loadGraph (RL, t_1, t_2)$\;
    $\gamma = findRoute(m,\mathcal{G})$;\ $\textit{//}$\footnotesize{Find route from $r_src$ to $r_dst$}\;
    \If{($\gamma \neq \phi$)}
    {
        $RL.insert$($\gamma$,$t_1$,$t_2$)\;
        Return $\gamma$ and [$t_1$,$t_2$]\;
    }
    \Else
    {
        $\textit{//}$\footnotesize{Find the next time point when a route is released}\;
        $t_1 = nextRelease(t_1,ReverseList)$\;
        $t_2 = t_1 + max(L_{w/oc})$\;
    }
}
\end{algorithm}

Finally, we analyze the time complexity of these two route algorithms.
For the first routing algorithm, the total number of loops equals $|\mathcal{R}|$.
Inside it, minimum finding with time complexity $\mathcal{O}(log|\mathcal{R}|)$ and cost computation with time complexity $\mathcal{O}(|\mathcal{M}|)$ are computed.
Finally, the time complexity is $\mathcal{O}(|\mathcal{R}|(log |\mathcal{R}|+|\mathcal{M}|))$.
For the second routing algorithm, the time complexity mainly depends on its route computation algorithm.
For the Dijkstra algorithm optimized by the binary heap, the time complexity is $\mathcal{O}(|\mathcal{R}|log|\mathcal{R}|)$.
Thus, our proposed algorithm can be solved in polynomial time.
Since NoC size is limited and route computation starts after task mapping and before transmission, it is feasible for route computation to be completed before data transmission.
However, the case in which route computation takes a longer time than task computation exists. 
The route computation would compute a default route, i.e., XY route in our design, at the beginning. 
In this case, the default simple XY route is returned as the result. 
We do not prove that our algorithms are deadlock free. 
However, even if the deadlock occurs during the link arbitration, the controller can stop it by integrating the deadlock detection algorithm \cite{mak2011embedded}.
We should note that, users can design their own route computation algorithm to minimize the delay caused by contention and maximize resource utilization. 





\section{Experimental Evaluation}
\label{section:5.2-experiment}
\subsection{Experimental Setup}
\label{implementation}
\begin{table}
    \caption{NoC Configurations}
    \renewcommand{\arraystretch}{1.2}
    \arrayrulewidth=0.85pt
    \tabcolsep 3pt
    \centering
    \footnotesize
    \begin{tabular}[c]{c c<{\centering} c<{\centering}}
        \hline
        \hline
        {\bf NoC Type}&\bf{Abbr.}&\bf{Description}\\
        \hline
        SMART &S&  The SMART NoC proposed in \cite{krishna2013breaking}. \\
        SSR-Net & SSR& The SMART NoC with pre-SSR Network \cite{chen2016reducing}.\\
        SHARP&SP& The propapataion-based SSR arbitration \cite{SHARPNoC}.\\
        ArSMART&A& Our proposed NoC.\\
        \hline
        
    \hline
    \end{tabular}
    \label{NoCArc}

\end{table}

To verify the advantages of our design, we conduct experiments to compare ArSMART NoC with SMART NoC and our algorithms with other routing solutions for SMART NoC. 
The NoC designs we considered in this article are shown in Table \ref{NoCArc}. Although novel designs revise the original SMART \cite{chen2016reducing,perez2019smart++,SHARPNoC} to reduce the overhead, \eg area and energy, latency is not improved or even worse than the original design. Thus, we set the original SMART as the baseline in experiments except for overhead analysis.   
For the routing algorithm, we choose the widely applicable XY routing and the start-of-the-art contention-aware routing algorithm proposed in \cite{9045103}. 
This routing algorithm is implemented at the design time and needs the support of 2D SSR. 
Also, additional overhead should be added but we consider the ideal case without overhead here. 
The configurations of routing algorithms are listed in Table \ref{routingAlgo}. 
The default mapping algorithm is the state-of-the-art contention-aware mapping proposed in \cite{yang2017task}.
\begin{table}
    \caption{Default Configurations of Routing Algorithm}
    \renewcommand{\arraystretch}{1.0}
    \arrayrulewidth=0.85pt
    \tabcolsep 3pt
    \centering
    \footnotesize
    \begin{tabular}[c]{c c<{\centering} c<{\centering}c<{\centering}}
        \hline
        \hline
        {\bf Routing Algorithm}&\bf{Abbr.}&\bf{Classification}&\bf{Required Info.}\\
        \hline
        XY &  xy & Deterministic& $r_{src},r_{dst}$\\
        Contention-aware\cite{9045103} & O & Deterministic&$r_{src},r_{dst},|m|,q$\\
        General Case& R1  & Adaptive&$r_{src},r_{dst},|m|$\\
        Time-triggered Case& R2 & Adaptive&$r_{src},r_{dst},|m|,q$\\
          \hline
        
    \hline
    \end{tabular}
    \label{routingAlgo}
\end{table}

Since no central controlled NoC simulator is available, to verify the functionality and correctness of the micro-architectural component, we implement our NoC design in Gem5 \cite{binkert2011gem5}, which provides a simulation kernel. 
As mentioned before, the router is redesigned and the cluster controller is functionally proposed. 
The main router and controller parts are depicted in Fig. \ref{figure:implementation}. 
The dependency between each component is represented by the arrow. 
The route of the message is computed by the controller. 
As long as the transmission data is prepared in the PE's network interface (NI), the controller begins to arbitrate links for messages. 
If all links are available, the controller would reserve the link for this message and send configuration information to corresponding routers. 
Also, the routers along the route would be woken up and configured after receiving the configuration from the controller. 
Then, data traverses the link from the source to the destination. 
Finally, the controller releases the resources after transmission completion and lets them be available for other transmissions. 

\begin{figure}
\centering
\includegraphics[width=3.5in]{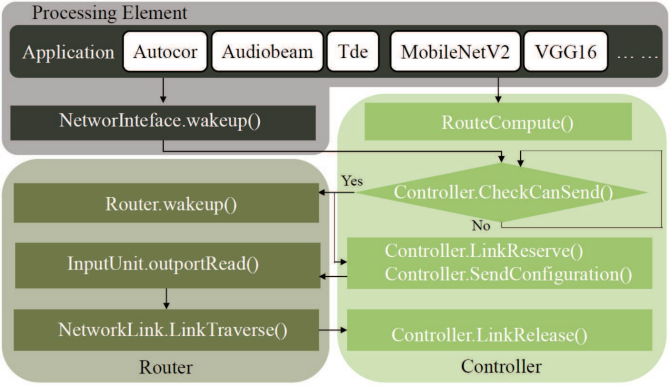}

\caption{ArSMART NoC implementation abstract}

\label{figure:implementation}
\end{figure}
\begin{table}
\label{table:paraSet}
    \caption{Default Configurations of NoC System}
    \renewcommand{\arraystretch}{1.2}
    \arrayrulewidth=0.85pt
    \tabcolsep 7pt
    \centering
    \footnotesize
    \begin{tabular}[c]{c c<{\centering} c<{\centering}}
        \hline
        \hline
        {\bf Parameter}&\bf{SMART NoC}& {\bf ArSMART NoC} \\
        \hline
        Topology &  2D Mesh & 2D Mesh \\
        NoC Size & $8 \times 8$ & $8 \times 8$ \\
        Cluster Size & - & $8 \times 8$ \\
         ${HPC}_{max}$ & 8  & 8 \\
        Flit Width  & 128-bit & 128-bit \\
        Package Size & 4 flits & 4 flits \\
        Buffer Size  & 4 flits & - \\
        VCs   & 2 VCs/port & - \\
        Router Pipeline & Three-stage & - \\
        Controller & - & 10MB\\
        L1 \& D Cache & Private, 32KB & Private, 32KB \\
        L2 Cache & Shared, 512KB/bank & Shared, 512KB/bank\\
        Frequency & 1 GHz & 1 GHz\\
        Technology & 22 nm & 22 nm\\

        \hline
        
    \hline
    \end{tabular}
    \label{parameter}
\end{table}

Our router class is derived from the original Gem5 router class, and thus, features of the Gem5 network are still available in our design, i.e., the network topology is configurable and interconnect bandwidth is changeable. 
However, since the intermediate router buffer is eliminated in our design, the buffer size configuration is disabled. The parameter setting for our experiments is shown in Table \ref{table:paraSet}. Gem5 records every event in every cycle. Combined with the energy consumption of each event simulated using Hspice in the 22-nm library, we accumulate the NoC energy consumption.

\vspace{-10pt}
\subsection{Evaluation Results for Real Applications}
Together with SMART NoC implemented using Gem5, we analyze ArSMART NoC performance using different metrics. Since the algorithm O and R2 need the information of execution time, we perform testing on a variety of streaming applications as well as AI applications, which have few branches and uncertainty during the execution.
Streaming application task graphs are generated from StreamIt benchmark \cite{thies2010empirical} and tasks are mapped using the algorithm proposed in \cite{yang2017task}. 
AI applications and their mappings are generated using Maestro \cite{kwon2018maestro}, which maximizes data reuse to decrease data movement. The cluster size for 4$\times$4, 8$\times$8 and 16$\times$16 NoC is 4$\times$4, 8$\times$8 and 8$\times$8 in this section, respectively.

\begin{figure}
\centering
\includegraphics[width=3.5in]{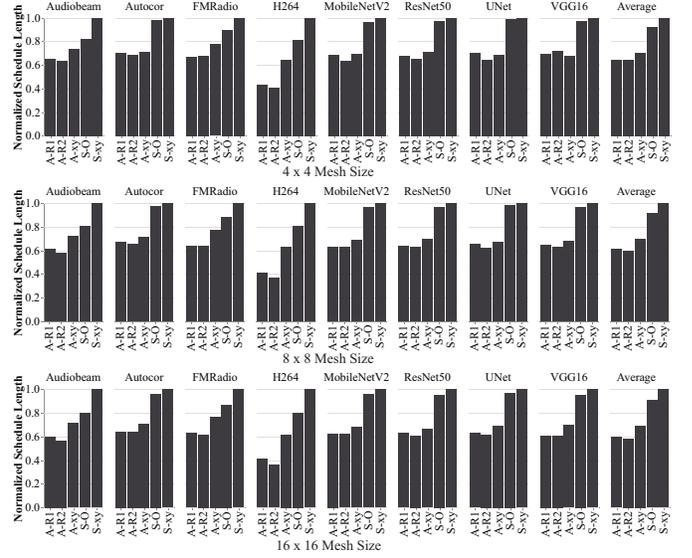}
\vspace{-20pt}
\caption{NoC performance comparison in terms of normalized total schedule length.}
\label{figure:runtime}

\end{figure}
\begin{figure}
\centering
\includegraphics[width=3.5in]{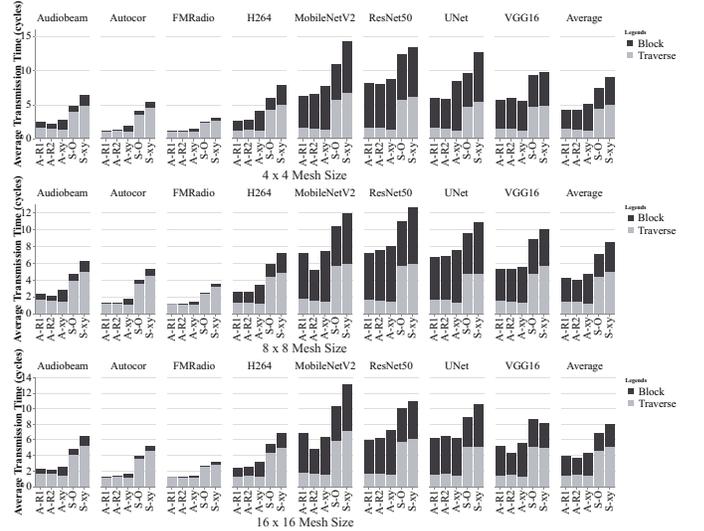}
\vspace{-20pt}
\caption{NoC performance comparison in terms of average transmission latency.}
\label{figure:stall}
\vspace{-10pt}
\end{figure}
\begin{figure}
\centering
\includegraphics[width=3.5in]{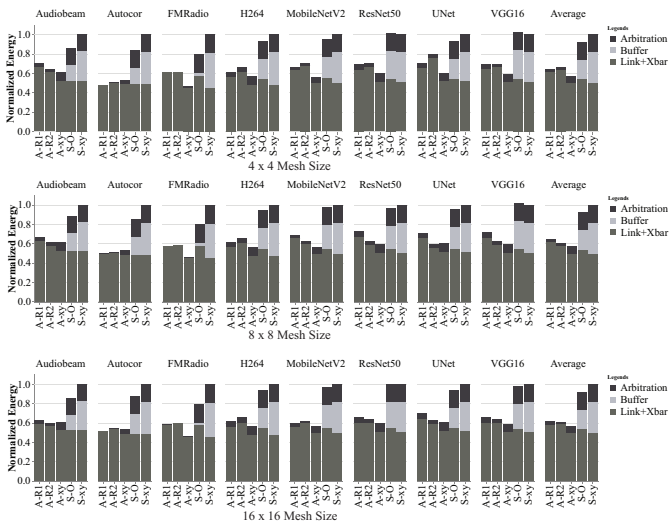}
\vspace{-20pt}
\caption{NoC performance comparison in terms of normalized total energy.}
\label{figure:energy}

\end{figure}

\begin{figure*}[t!]
\centering
\includegraphics[width=7.0  in]{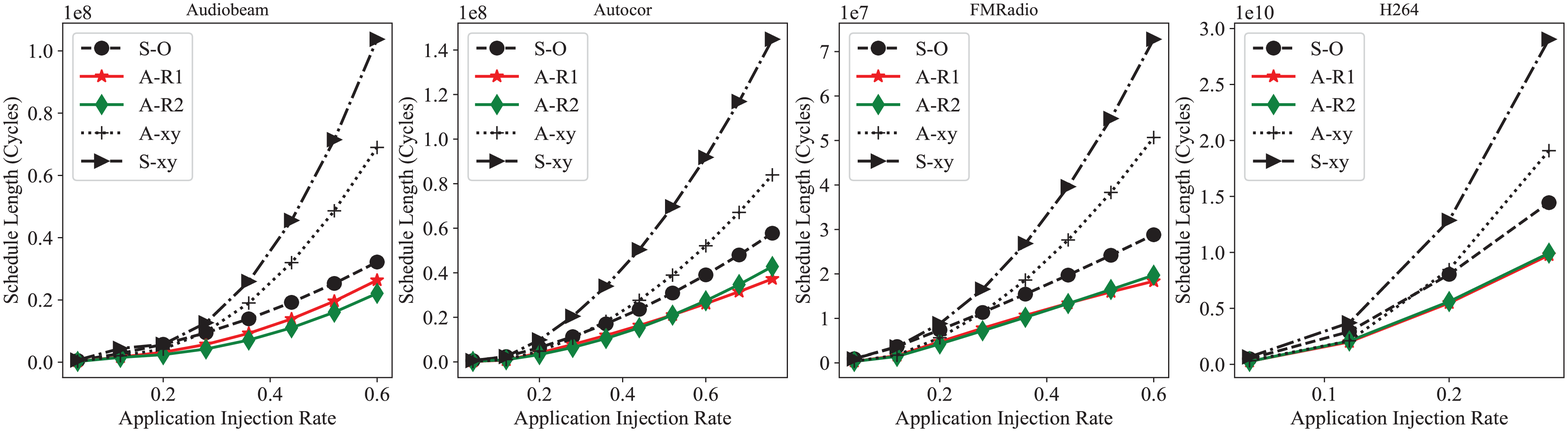}
\caption{Schedule length comparison with different application injection rates.}
\label{figure:AIR}
\vspace{-15pt}
\end{figure*}

\noindent  $\bullet$ \textbf{Schedule Length.} Fig. \ref{figure:runtime} shows the schedule length comparison for different applications. 
On average, our approach reduces 34.1\%, 39.2\% and 40.7\% total schedule length over SMART NoC for 4$\times$4, 8$\times$8 and 16$\times$16 NoC size using R1.
The performance of A-R1 and A-R2 are very close, indicating that our general algorithm is also efficient but can be used in a wider range. 
Meanwhile, our method dramatically reduces ${L}_{conf}$ and ${L}_{cs}$, resulting in performance improvement. 
With the increment of NoC size, such the improvement becomes more obvious since that more link resources can be used to establish a long-range path.
By comparing the results of S-xy and A-xy, since these two different NoCs use the same routing algorithm, we can roughly get the hardware improvement. 
We should note that such improvement is not very precise since the scheduling order changes for different cases. 
The remaining improvement is contributed by the routing algorithm.
Our technique outperforms SMART NoCs slightly on AI applications, relatively, whose task graphs are too complicated. However, thanks to the quick configuration process, we achieve at least 12.2\% of latency reduction over SMART NoC.

\noindent  $\bullet$ \textbf{Average network latency.} In order to break down the performance gain, we conduct experiments for average network latency shown in Fig. \ref{figure:stall}. 
The trend of average network latency is similar to schedule length. 
Generally, ${L}_{conf}$ is just a few clock cycles and only counts in message level rather than package level. 
For the contention delay, our method eliminates the contention at intermediate routers. 
An additional delay at the source ${L}_{cs}$ is considered. Using our routing strategies, ${L}_{cs}$ can be decreased. 
By comparing the blocking time for A-O, S-R1 and S-R2, we conclude that both of our algorithms are efficient.  

\noindent  $\bullet$ \textbf{Energy.} Energy consumption, presented in Fig. \ref{figure:energy}, of our NoC is much less than SMART NoCs using XY routing. 
On average, our design reduces 25.3\%, 27.4\% and 29.7\% energy consumption over SMART NoC for 4$\times$4, 8$\times$8 and 16$\times$16 NoC size, using R1. 
The energy deduction mainly comes from the removed buffering and decreased arbitration in message level rather than package level. 
However, since our routing algorithms adopt the arbitrary-turn route which may be longer than the route XY routing chooses, the A-R1 and A-R2 consume more energy on link and crossbar traversal than A-xy and S-xy.

\noindent  $\bullet$ \textbf{AIR.} With the increment of application injection rate (AIR), shown in Fig. \ref{figure:AIR}, the application schedule length is increased. 
We can observe that the increment of total schedule length is relatively slow using our approach. 
In particular, in the case of heavy congestion, performance improvement is obvious for our proposed routing strategies. 
This shows that our routing strategies always select a route with less contention according to network state.

\vspace{-10pt}
\subsection{Evaluation Results for Synthetic Traffics}
To further explore the advantages and limitations of ArSMART NoC, we generate several random task graphs. Default parameters for these task graphs are listed in Table \ref{taskParameter}. 
\begin{table}
    \caption{Default Parameter Settings For Synthetic Traffics}
    \renewcommand{\arraystretch}{1.2}
    \arrayrulewidth=0.85pt
    \tabcolsep 7pt
    \centering
    \footnotesize
    \begin{tabular}[c]{c c<{\centering} c<{\centering} c<{\centering}}
        \hline
        \hline
        {\bf Parameter}&\bf{Default} &\bf{Parameter}& {\bf Default} \\
        \hline
        Number of Nodes & 100 & Number of Links & 300 \\
        
        Avg. Task Volume & 8192 & Avg. Message Size & 8192 \\
        Heterogeneity Degree & 1 & Mesh Size & $8 \times 8$\\
        Package Size  & 10 flits & Mapping Algorithm & \cite{yang2017task} \\
        \hline
        
    \hline
    \end{tabular}
    \label{taskParameter}
\vspace{-10pt}
\end{table}

\begin{figure}
\centering
\includegraphics[width=3.5in]{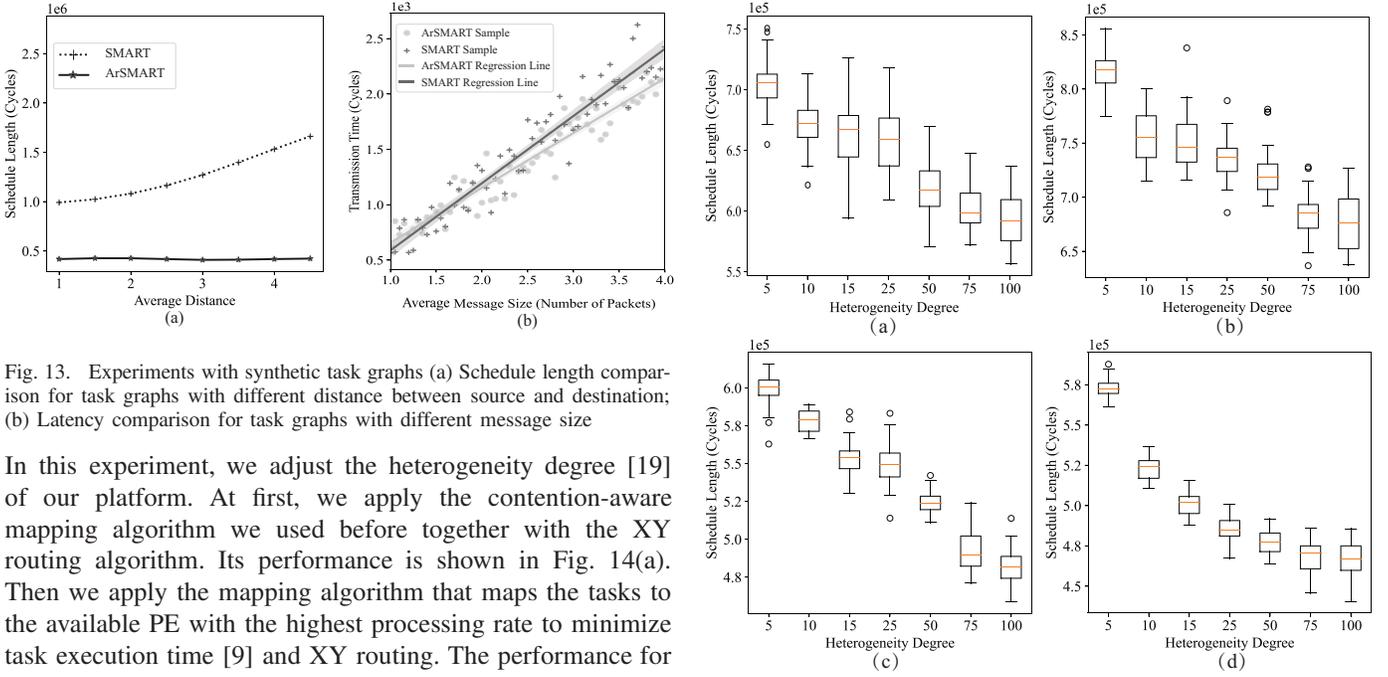}
\caption{Experiments with synthetic task graphs (a) Schedule length comparison for task graphs with different distance between source and destination; (b) Latency comparison for task graphs with different message size}
\vspace{-15pt}
\label{figure:hardware}

\end{figure}
\noindent  $\bullet$ \textbf{Distance.} In this experiment, we manipulate the task mapping algorithm to change the average distance between the source and the destination PEs from 1 to 5.
Fig. \ref{figure:hardware}(a) shows our NoC is scarcely influenced by the distance while SMART NoC latency notably increases. 
The reason for this result may be that the longer path SMART NoC wants to establish, the more possible this transmission would encounter the interruption by other messages. 
We can conclude that our proposed design is not distance sensitive, which allows the transmission with long-distance to have the same performance as the transmission between adjacent routers. 

\noindent  $\bullet$ \textbf{Message Size.} ArSMART NoC establishes the path in the message level rather than packet level so that the overhead for path establishment is amortized. 
Thus, the configuration time saving is more obvious for large message sizes. 
However, for the SMART NoC, the path configuration is conducted at the packet level, the configuration overhead is consistent for each packet, which means that the SMART NoC can process short messages more efficient than ArSMART NoC. 
To explore the message size with whom the ArSMART NoC can outperform the SMART NoC, in this experiment, we change the number of packets for each message and make the average message size change from 1 to 4.
Fig. \ref{figure:hardware}(b) shows the results for task graphs with task execution time equals to 1 and message size varies from 1 to 4. 
The regression result shows that our ArSMART NoC has a better performance for task graphs with the average message size larger than 1.67. 
However, we should note that this is not an accurate conclusion due to the influence of schedule orders.

\noindent  $\bullet$ \textbf{Heterogeneity Degree.} In previous experiments, the PEs are homogeneous and their processing rate are the same. 
For such platform, the contention-aware mapping algorithm is efficient since the objective can be transmission time minimization only. 
However, in the heterogeneous system, the problem becomes complex since the execution time should be taken into consideration, too. 
In such case, using mapping algorithms to optimize execution and routing to minimize the transmission is an efficient method to decrease overall schedule length. 
In this experiment, we adjust the heterogeneity degree \cite{ali2000task} of our platform. 
At first, we apply the contention-aware mapping algorithm we used before together with the XY routing algorithm. 
Its performance is shown in Fig. \ref{figure:software}(a). 
Then we apply the mapping algorithm that maps the tasks to the available PE with the highest processing rate to minimize task execution time \cite{holzenspies2010run} and XY routing. 
The performance for this case is presented in Fig. \ref{figure:software}(b). 
Fig. \ref{figure:software}(c) and (d) use the same mapping algorithm as (b) but apply our proposed routing algorithm R1 and R2. 
With the increment of heterogeneity degree, the schedule length for all cases decreases, due to the average processing rate of all PEs increases. 
However, the last two cases have less schedule length compared with the former two cases. 
Also, the divergence of the case (c) and (d) are less. 
This shows that the arbitrary-turn route has stable performance improvement. Through this experiment, we emphasize the importance of using arbitrary-turn route to optimize transmission performance. 

\begin{figure}
\centering
\includegraphics[width=3.5in]{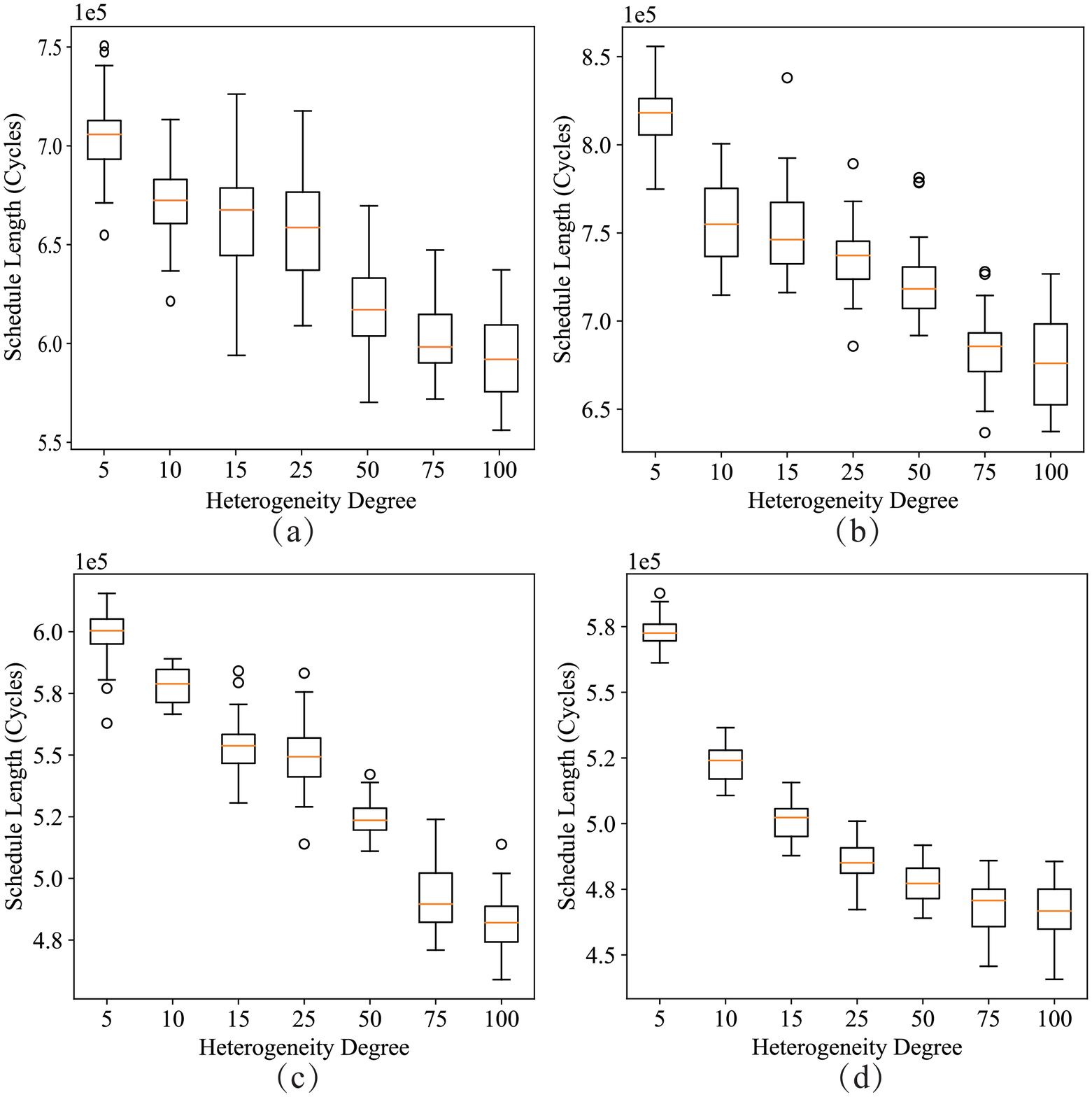}
\caption{Performance analysis with different task graphs (a) Contention-aware mapping with XY routing; (b) Computation-aware mapping with XY routing; (c) Computation-aware mapping with R1 routing; (d) Computation-aware mapping with R2 routing. }
\label{figure:software}
\vspace{-15pt}
\end{figure}

\vspace{-10pt}
\subsection{Overhead Analysis}
In this section, we compare ArSMART to original SMART and improved SMART designs, SSR-Net \cite{chen2016reducing} and SHARP NoC \cite{SHARPNoC}, in terms of area and power. All of these NoCs apply the 2D configuration. The SMART suffers from high overhead issues since each router must consider all SSRs from upstream routers. At most $\textit{HPC}_{max}(2 \textit{HPC}_{max} -1)$ SSRs at each input port are needed, which largely increases wire and arbitration logic area as well as power consumption. SSR-Net, in which an auxiliary SSR network is used, significantly reduces the wire area. Based on this design, SHARP NoC is proposed to eliminate the quadratic arbitration by the propagation-based SSR arbitration mechanism. The tool we used, called DSENT \cite{dsent}, for router area and power estimation is the same as \cite{chen2016reducing,SHARPNoC} for a fair comparison. The technology class we used in DSENT is 22nm. The storage overhead for SMART, SHARP and SSR is the overhead for the buffer, while for ArSMART is the delay register. Since the message arbitration is performed in the controller rather than the router, the arbitration overhead of ArSMART includes the decoding and configuration overhead.

\noindent  $\bullet$ \textbf{Power.} 
In \fig \ref{figure:area}(a) the router dynamic power for ArSMART, SHARP, SMART and SSR with traffic load of low, medium and high are shown. In this experiment, the $\textit{HPC}_{max} = 6$ is applied since this is the best achievable values for SHARP under our proposed system configuration. When the traffic load is low, ArSMART reduces about 68\% of the power consumption compared with the original SMART design. With the increment of traffic load, the buffer and arbitration power consumption becomes more dominating. Thanks to our ``blocking at source'' design, the power of storage is greatly reduced. The simplified router design helps ArSMART reduce the arbitration overhead. Under the high traffic load condition, our design can reduce about 70\% of power over SMART NoC.

\noindent  $\bullet$ \textbf{Area.} Finally, the area overhead is analyzed. In \fig \ref{figure:area}(b) the router area for SHARP, ArSMART, SMART and SSR with respect to $\textit{HPC}_{max}=2,4,6,8$ are shown. Since up to $\textit{HPC}_{max}(2 \textit{HPC}_{max} -1)$ SSRs need to be arbitrated in every port for SMART and SSR design, the quadratic increment in arbitration area for these two cases is observed accordingly. Although SHARP decreases arbitration overhead obviously, it cannot support high $\textit{HPC}_{max}$. The maximum $\textit{HPC}_{max}$ SHARP support is only 6. Since in Our ArSMART, the control plane and data plane are separated, the arbitration control is relatively low compared with the other distributed SMART designs (\ie SMART, SSR, SHARP). Also, our arbitration overhead is not scaled up with the increment of $\textit{HPC}_{max}$. Moreover, since in our design, only one delay register is needed for each port, the buffer area is also decreased. The router area overhead of ArSMART is 2.6x-6.8x the original SMART.

\begin{figure}
\centering
\includegraphics[width=3.5in]{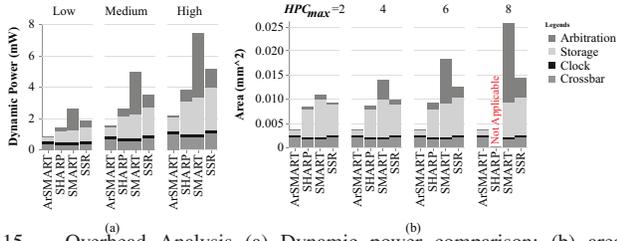}
\vspace{-25pt}
\caption{Overhead Analysis (a) Dynamic power comparison; (b) area comparison.}
\label{figure:area}
\vspace{-15pt}
\end{figure}

\vspace{-10pt}
\section{Related Works}
\label{section:2-relatedwork}
In this section, we discuss related NoCs designs to reduce latency, improve energy efficiency or decrease area overhead.

\noindent  $\bullet$ \textbf{SMART NoC.} 
SMART NoC \cite{krishna2013breaking,chen2013smart,krishna2013single} is proposed to reduce the end-to-end latency by enabling single-cycle multi-hop traversal.
In \cite{chen2016reducing}, a control network was proposed to reduce wire and energy overhead of the original SMART NoC. 
Generally, this solution reduces ${L}_{head}$ but ${L}_{ct}$ has not been solved effectively. 
Moreover, extra arbitration may be added for bypassing signals. 
The most advanced SMART design, SHARP NoC \cite{SHARPNoC}, is proposed to eliminate the quadratic arbitration by the propagation-based SSR arbitration mechanism.   

\noindent  $\bullet$ \textbf{SDNoC.} 
To support arbitrary-turn routing which has no constraints on routing decision, software-defined NoC (SDNoC) design has been proposed in \cite{cong2014configurable}. 
However, these approaches focus on updating the flow table of the router, which increases the complexity of the arbitration dramatically. 
Finally, it may decrease ${L}_{ct}$ by adopting excellent contention minimize algorithm we mentioned before, but $t_r$ increased so that the overall performance is degraded. 
Software-defined circuit switching NoC \cite{ruaro2017sdn,ruaro2019distributed,ruaro2018software,ellinidou2018sdn} uses the controller to configure the transmission route. 
However, these researches use additional packet switching to transmit data packets when there is no path in SDNoC and apply hop-by-hop configuration and transmission, adding additional overhead to this design.

\noindent  $\bullet$ \textbf{Bufferless NoC.} 
On-chip router buffers put pressure on the area and power constraints for NoC. 
Research \cite{hoskote20075} shows that 22\% of router power is consumed by network buffering resources. \cite{moscibroda2009case} describes new algorithms to route packets without buffering.
By controlling the injection rate and deflecting flits to undesired ports, buffers can be eliminated. 
In \cite{hayenga2009scarab}, SCARAB shows a single-cycle bufferless router design. 
Together with a processor-side buffered router, it reduces the possibility of packet drops and re-transmission costs. 
Improved bufferless router designs \cite{picornell2019dcfnoc,wang2019surf} have been proposed. 
However, since no buffer can temporarily hold packets, packets in bufferless-routing have to keep moving in the links. This may cause additional latency overhead.

\noindent  $\bullet$ \textbf{Application-specific NoC.}  
Application-specific NoCs \cite{jackson2010skip,modarressi2010virtual,sewell2012swizzle} generate NoCs in accordance with the application’s communication graph that is known apriori. 
For AI applications, the novel NoC design \cite{kwon2017rethinking} is proposed to boost communication performance for spatial neural network accelerator is presented.
However, fixed NoC architecture cannot benefit all kinds of applications.
Such configurable NoC does not support dynamic change in the run-time, which means the transmission pattern change cannot be handled in this model. 
However, the traffic during the run-time is not static; it varies phase by phase and is dependent on the mapping of the dataflow over PEs, and the input parameters.

The aforementioned approaches use different technologies and have their specific benefits, as we list in Table \ref{comparison}.
We compare them in four metrics, low latency, low cost (i.e., power and area), high adaptability and high adaptability. As we can see from the table, these techniques cannot dominate with each other regarding these costs and benefits. 
Bypassing intermediate routers (i.e., SMART NoC) and applying routing algorithms with arbitrary-turn (i.e., SDNoC), ``A-route'' for simplicity, to get low latency needs additional control network, which let the energy consumption increased. 
The methods with low hardware cost (i.e., Bufferless NoC) cannot meet the latency constraints. 
The application-specific NoC, denoted as ``App. Specific'' in the table, which designs NoC by adding additional links at the design time is hard to develop and adapt for software updates. 
Also, this method cannot handle the transmission pattern change during the run-time.
\begin{table}
    \caption{NoC Comparison}
    \renewcommand{\arraystretch}{1.2}
    \arrayrulewidth=0.85pt
    \tabcolsep 3.3pt
    \centering
    \footnotesize
    \label{comparison}
    \begin{tabular}{ p{2.1cm} c<{\centering}c<{\centering} p{0.6cm}<{\centering}  c<{\centering} c<{\centering}}
        \hline
        \hline
         \multirow{2}{*}{{ \bf NoC}} & \multicolumn{2}{c}{\bf Low latency}
          & \multirow{2}{*}{ \makecell[c]{{\bf Low }\\ {\bf Cost}}}  &  \multirow{2}{*}{\makecell[c]{{\bf High }\\ {\bf Adaptability}} }& \multirow{2}{*}{\makecell[c]{{\bf High }\\ {\bf Generality}} }\\
          ~ & {\bf Bypass} & {\bf A-route} & ~ & ~ & ~\\
          
        \hline
        Traditional  & $\times$  & $\times$  & $\times$  & $\checkmark$& $\checkmark$\\
        SMART & $\checkmark$  & $\times$  & $\times$  & $\checkmark$& $\checkmark$  \\
        SDN  & $\times$  & $\checkmark$  & $\times$  & $\checkmark$ & $\checkmark$\\
        Bufferless & $\times$  & $\times$  &  $\checkmark$  & $\checkmark$ & $\checkmark$\\
        App. Specific & $\checkmark$  & $\times$  &  $\checkmark$ &$\times$ &$\times$\\
        { \bf Our ArSMART} & {\bf $\checkmark$}  & $\checkmark$  & $\checkmark$  & $\checkmark$ &$\times$\\
        \hline
        
    \hline
    \end{tabular}
    
        \vspace{-15pt}
\end{table}

\section{Conclusion}

\label{section:6-conclusion}
In this article, we proposed an NoC design, ArSMART NoC, which supports single-cycle long-distance data transmission among many cores. Using a cluster-based control method, ArSMART NoC supports any arbitrary-turn route which can decrease contentions.
By configuring the routers directly, we presented a method to setup arbitrary-turn routes from a source to a destination within a very small number of cycles. 
We also introduced routing algorithms to further decrease communication contention. 
Compared with the high-performance SMART NoC, we have decreased 40.7\% application schedule length and 29.7\% in energy consumption on average. Thanks to the simplified buffer and arbitration components, our design reduced the router area and power consumption compared with the start-of-art overhead-aware SMART NoC designs. In the current work, we demonstrate the software design of the controller. Considering the high energy efficiency and performance provided by ASIC, we plan to design specific hardware to implement and accelerate the function of the ArSMART controller.

\section*{Acknowledgments}
This work is partially supported by the Ministry of Education, Singapore, under its Academic Research Fund Tier 2 (MoE2019-T2-1-071) and Tier 1 (MoE2019-T1-001-072), and Nanyang Technological University, Singapore, under its NAP (M4082282) and SUG (M4082087).

\bibliographystyle{IEEEtran}
\bibliography{IEEEabrv,reference}

\end{document}